\documentclass[11pt]{article}
\usepackage[
bookmarks=true,
bookmarksnumbered=true,
bookmarksopen=true,
pdfborder={0 0 0},
breaklinks=true,
colorlinks=true,
linkcolor=black,
citecolor=black,
filecolor=black,
urlcolor=black,
]{hyperref}

\usepackage[all=normal,paragraphs,wordspacing,tracking]{savetrees}
\usepackage{times}
\usepackage[text={6.9in,9in},centering]{geometry}
\usepackage{setspace}
\usepackage{amsthm}
\usepackage{amsmath}
\usepackage{amsfonts}
\usepackage{amssymb}
\usepackage{graphicx}
\usepackage{pdfpages}
\usepackage{stmaryrd}
\usepackage{enumitem}
\usepackage[capitalise]{cleveref}
\usepackage{nicefrac}
\usepackage{verbatim}
\usepackage[normalem]{ulem}
\usepackage{verbatimbox} 
\usepackage{bm}
\usepackage{caption}
\usepackage{subcaption}

\newtheorem{theorem}{Theorem}
\newtheorem{definition}{Definition}
\newtheorem{lemma}{Lemma}
\newtheorem{corollary}{Corollary}

\newtheorem*{notation*}{Notation}
\newtheorem{claim}{Claim}

\newcommand{\act}{\alpha}
\newcommand{\RP}{R_{\cal P}}

\newcommand{\decision}{\mathsf{dec}}
\newcommand{\abort}{\mathsf{abort}}
\newcommand{\commit}{\mathsf{commit}}

\newcommand{\AllOne}{\mathtt{all[1]}}

\newcommand{\err}{$\overline{\mbox{\textit{err}}}$}
\newcommand{\Proc}{\mathbb{P}}
\newcommand{\bbF}{\mathbb{F}}

\newcommand{\defemph}[1]{\textbf{\textit{#1}}}

\newcommand{\modelf}{\gamma^f}

\newcommand{\BB}{\mathtt{B.1.Consensus}}
\newcommand{\Sr}{S^r}

\renewcommand{\cref}{\Cref}
\newcommand{\ind}[2]{\approx_{#1}^{#2}}
\newcommand{\Sone}{{\bf S1}}
\newcommand{\Stwo}{{\bf S2}}
\newcommand{\Sthree}{{\bf S3}}
\newcommand{\Prot}{{\cal P}}
\newcommand{\mcc}[1]{\leadsto_{#1}}%
\newcommand{\mccr}{\mcc{r}}
\newcommand{\eqdef}{\triangleq}
\newcommand{\Stealth}{\mbox{{\sc Stealth}}}
\newcommand{\Quick}{\mbox{{\sc D2}}}
\newcommand{\Short}{\mbox{{\sc D1}\rm{f1}}}
\newcounter{linecounter} 
\newcommand{\linenumbering}{\ifthenelse{\value{linecounter}<10} 
	{(0\arabic{linecounter})}{(\arabic{linecounter})}}

\renewcommand{\thelinecounter}{\ifnum \value{linecounter} >  
	9\else 0\fi \arabic{linecounter}}

\newlength {\squarewidth}

\title{Silence}
\author{Guy Goren\\
Viterbi Faculty of Electrical Engineering,\\ Technion\\
\texttt{sgoren@campus.technion.ac.il}
\and 
Yoram Moses\\
Viterbi Faculty of Electrical Engineering,\\ Technion\\
\texttt{moses@ee.technion.ac.il}  
}
\begin{document}
\date{}
\maketitle

\begin{abstract}

The cost of communication is a substantial factor affecting the scalability of many distributed applications. Every message sent can incur a cost in storage, computation, energy and bandwidth. Consequently, reducing the communication costs of distributed applications is highly desirable. The best way to reduce message costs is by communicating without sending any messages whatsoever. 
This paper initiates a rigorous investigation into the use of silence 
in synchronous settings, in which  processes can fail. 
We formalize sufficient conditions for information transfer using silence, as well as necessary conditions for particular cases of interest. 
This allows us to identify message patterns that enable communication through silence. 
In particular, a pattern called a {\em silent choir} is identified, and shown to be central to information transfer via silence in failure-prone systems. 
The power of the new framework is demonstrated on the {\em atomic commitment} problem (AC). 
A complete characterization of the tradeoff between message complexity and round complexity in the synchronous model with crash failures is provided, in terms of lower bounds and matching protocols. 
In particular, a new message-optimal AC protocol is designed using silence, in which processes decide in~3 rounds in the common case. 
This significantly improves on the best previously known message-optimal AC protocol, in which decisions were performed in $\Theta(n)$ rounds. 
\end{abstract}
\bigskip\bigskip

$~$\hfill\begin{tabular}{ll}
\textit{And in the naked light I saw} \\[.5ex]
\textit{Ten thousand people, maybe more}\\[.5ex]
\textit{People talking without speaking}\\
$~$$\cdots$\\
\textit{And no one dared}\\[.5ex]
\textit{Disturb the sound of silence}\\[1ex]
&\hspace{-1cm}{Paul Simon, 1964}
\end{tabular}

\medskip\vfill
\vfill

\thispagestyle{empty}
\newpage
\setcounter{page}{1}
\section{Introduction}\label{sec:intro}
The cost of communication is a substantial factor limiting the scalability of many distributed applications (see, e.g. \cite{Bitcoin,zookeeper,XFT,jogalekar2000evaluating, lev2016modular}). Indeed, sending a message imposes costs in storage, computation, energy and bandwidth. Consequently, reducing the communication costs of distributed applications is highly desirable. The best way to reduce these costs is by communicating without sending any messages whatsoever. 
In reliable synchronous systems processes can exchange information in silence, effectively ``{\it sending a message by not sending a message},'' to use a term from~\cite{lamport1984using}.  In fault-prone systems, however, using silence in this way is considerably more subtle. Roughly speaking, in a reliable system in which messages from~$j$ to~$i$ are guaranteed to be delivered within~$\Delta$ time units, if~$j$ sends no message to~$i$ at time~$t$, then~$i$ is able to detect this fact at time~$t+\Delta$. 
This, in turn, can be used to pass information from~$j$ to~$i$. 
In a setting in which~$j$ may fail, however, it is possible for~$i$ not to receive~$j$'s message because~$j$ failed in some way.

Given the value of reducing communication costs and the power of silence to do so, efficient protocols for a variety of tasks of interest
make effective, albeit implicit, use of silence 
	(see, e.g., \cite{bar1991consensus, amdur1992message, hadzilacos1993message, guerraoui2017fast, liskov1993practical, dolev1983authenticated, bernstein1985loosely}). 
There is no clear theory underlying this use of silence, however. 
This paper initiates a rigorous investigation into the use of silence in synchronous settings in which  processes can fail. 
We formalize sufficient conditions for information transfer using silence, as well as necessary conditions for particular cases of interest. 
This allows us to identify message patterns that enable communication through silence. 

 Using the proposed framework, we consider the {\em atomic commitment} problem (AC) in synchronous systems with crash failures \cite{BHG87,dwork1983inherent} as a case study. A silence-based analysis provides new lower bounds 
 on the efficiency of AC protocols that are optimal in the common case. The conditions identified for using silence are then used to provide protocols that match the lower bounds. In some cases, the new protocols significantly improve the state-of-the-art.

The main contributions of this paper are: 
	\begin{itemize}[itemsep=0pt]
		\item Simple rules for silent information transfer in the presence of failures are provided and formalized. 
		\item The {\em Silent Choir} Theorem, which precisely captures the communication patterns required for 
		silence to convey information about the initial values, is proven. 
		In failure-prone systems, a silent choir plays the role that null messages play in reliable systems, 
		facilitating communication via silence. 
		\item Using a silence-based analysis,  lower bounds are established on the number of messages needed for  AC protocols to decide in~$D$ rounds in the common case, for any given~$D\ge 1$. 
		\item Silent information transfer rules are used to design protocols that prove these bounds to be tight, 
		thereby completely characterizing the tradeoff between the message complexity and round complexity of AC. 
		 For the most interesting case of message-optimal AC, we present the \Stealth\/ protocol, which reduces the decision times from $\Theta(n)$  to $\Theta(1)$ rounds (in fact, from $n+2f$  to~3). 
	\end{itemize}
	
This paper is organized as follows. The next section introduces our model of computation, and formal definitions of indistinguishability, knowledge and AC. 
		\cref{sec:Silence} formalizes how silence can be used to convey information in fault-prone synchronous message-passing systems. 
		In \cref{sec:AC} lower bounds and upper bounds for optimal AC in the common case are established. Finally, concluding remarks are discussed in~\cref{sec:discussion}.



	\section{Definitions~and~Preliminary Results}\label{sec:def}
	\subsection{Atomic Commitment}	\label{sec:Ac-def}
	In the atomic commitment problem \cite{BHG87,dwork1983inherent, guerraoui2017fast,hadzilacos1987knowledge}, each process~$j$ starts out with a binary initial value $v_j\in\{0,1\}$.  Processes need to decide among two actions, $\commit$ or $\abort$. All runs of an AC protocol are required to satisfy the 
		following conditions: 

\newpage
	\begin{spacing}{.3}

		\noindent
		\underline{
			{\sc Atomic Commitment:}}
		\begin{itemize}[itemsep=.5pt]
			\item[]{\bf Decision:}\quad Every correct process must eventually decide either $\commit$ or $\abort$, 
			\item[]{\bf Agreement:}\quad All  processes that decide make the same decision, 
			\item[]{\bf Commit Validity}:\quad A process can decide $\commit$ only if all initial values are~$1$, and 
			\item[]{\bf Abort Validity}:\quad A process can decide $\abort$ only if  some initial value is~0, or if some process failed.\\[.3ex]
		\end{itemize}
	\end{spacing}
	\vspace{.5ex}
	\noindent
	The AC problem is motivated by distributed databases, in which the processes are involved in performing a joint transaction. An initial value of~0 can correspond to a vote to abort the transaction, while~1 is a vote to commit the transaction.  
	Typically, processes enter the AC protocol with a vote to commit, and failures in runs of AC are rare. We call a run in which all initial values are~1 and no failures occur, a {\em nice} run. As first advocated by~\cite{dwork1983inherent}, it is sensible to seek AC protocols that are optimized for the common case, i.e., for nice runs. Natural parameters to optimize are the number of messages sent in nice runs, and the number of rounds required for decision, and for halting. We call an AC protocol that sends the fewest possible messages in nice runs {\em message optimal}, and one that decides the soonest {\em round optimal}. 
	In a recent paper \cite{guerraoui2017fast}, Guerraoui and Wang present a message-optimal AC protocol. It uses $n+f-1$ messages in nice runs (for a system with~$n$ processes and up to~$f$ crash failures). But this message efficiency seems to come at a cost in terms of decision time: processes decide in nice runs only at the end of $n+2f$ rounds. This raises an interesting open problem of identifying the tradeoff between message complexity and decision time, for AC protocols in the common case.
	
%
%

\subsection{Model of Computation}
\label{sec:model}
	We consider the standard synchronous message-passing model with benign crash failures. We assume a set $\Proc=\{0,1,\ldots,n-1\}$ of $n> 2$ processes.
	%
	Each pair of processes is connected by a two-way communication link, and for each message the receiver knows the identity of the sender.
	All processes share a discrete global clock that starts at time~$0$ and  advances by increments of one.
	Communication in the system proceeds in a sequence of \emph{rounds}, with round~$m+1$ taking place between time~$m$ and time~$m+1$, for $m\ge 0$. A message sent at time~$m$ (i.e., in round~$m+1$) from a process~$i$ to~$j$ will reach~$j$ by time~$m+1$, i.e., the end of round~$m+1$. 
	In every round, each process performs local computations, sends a set of messages to other processes, and finally receives messages sent to it by other processes during the same round.
	A faulty process in a given execution fails by \emph{crashing} in some round~$m\ge 1$. 
	In this case, the process behaves correctly in the first~$m-1$ rounds and it performs no actions and sends no messages from round~$m+1$ on. 
	During its crashing round~$m$, the process sends messages to an arbitrary subset among the processes to whom the protocol prescribes that it send messages in round~$m$. It does not take any decisions at the end of its crashing round (i.e., at time~$m$). As is customary in agreement problems such as Consensus, we shall assume that each process~$i$ starts at time~$0$ in some \emph{initial state}, which can be assumed for simplicity in this paper to consist of its \emph{initial value} $v_i\in\{0,1\}$. 
	
	At any given time~$m\ge 0$, a process is in a well-defined  \defemph{local state}.   For simplicity, we assume that the local state of a process~$i$  consists of its initial value~$v_i$, the current time~$m$, and the sequence of the actions that~$i$ has performed (including the messages it has sent and received) up to that time. In particular, its local state at time~0 has the form~$(v_i,0,\{ \} )$. We will also assume that once a process has crashed, its local state becomes~$\bot$. A \defemph{protocol} describes what messages a process sends and what decisions it takes, as a deterministic function of its local state. 
	
	We will consider the design of protocols that are required to withstand up to~$f$ crashes.
	Thus, given $1\le f<n$, we denote by $\modelf$ the model described above in which it is guaranteed that no more than~$f$ processes fail in any given run. 
	We assume that a protocol~$\cal P$ has access to the values of~$n$ and~$f$, typically passed to~$\cal P$ as parameters.
	
	A \defemph{run} is a description of a (possibly infinite) execution of the system.
	We call a set of runs~$R$ a \defemph{system}.
	We will be interested in systems of the form $R_{\cal P}=R({\cal P},\modelf )$ consisting of all runs of a given protocol~$\cal P$ in which no more than~$f$ processes fail. 
	Observe that a protocol~$\cal P$ solves AC in the model~$\modelf$ if and only if every run of~$R_{\cal P}$  satisfies the Decision, Agreement and two Validity conditions described above. 
	Given a run~$r$ and a time~$m$, we denote the local state of process~$i$ at time~$m$ in run~$r$ by $r_i(m)$.
	Notice that a process~$i$ can be in the same local state in different runs of the same protocol. Since the current time~$m$  is represented in the local state $r_i(m)$, however,  $r(m)=r'(m')$ can hold only if $m=m'$.

\subsection{Indistinguishability and Knowledge}

We shall say that two runs~$r$ and~$r'$ are {\em indistinguishable} to process~$i$ at time~$m$ if $r_i(m)=r'_i(m)$. We denote this by $r\ind{i}{m}r'$. Notice that since we assume that processes follow deterministic protocols, if $r\ind{i}{m}r'$ then process~$i$ is guaranteed to perform the same actions at time $m$ in both~$r$ and~$r'$. 
Problem specifications  typically impose restrictions on actions, based on properties of the run. Moreover, since the actions that a process performs are a function of its local state,  the restrictions can depend on properties of other runs as well. 

For example, the Agreement condition in AC implies that a process~$i$ cannot decide to $\commit$ at time~$m$  in a run~$r$ if there is an indistinguishable run~$r'\ind{i}{m}r$ in which some process decides $\abort$. Similarly, by Commit Validity a process~$i$ can not $\commit$ if there is a run~$r'$ that is indistinguishable  from~$r$ (to~$i$ at time~$m$) in which one of the initial values is~0. Similarly, by Abort Validity~$i$ can not $\abort$ at time~$m$ if there exists a nice run~$r'\ind{i}{m}r$  (i.e., all initial values are~1 and no failure occurs in~$r'$). 
These examples illustrate how indistinguishability can inhibit actions --- performing an action can be prohibited because of what may be true at indistinguishable runs. 

Rather than considering when actions are prohibited, we can choose to consider what is required in order for an action to be allowed by the specification. To this end, we can view Commit Validity as implying that process~$i$ is allowed to perform $\commit$ at time~$m$ in~$r$ only if all initial values are~1 in every run~$r'\ind{i}{m}r$. 
This is much stronger than stating that all values are~1 only in the run~$r$ itself, of course. Roughly speaking, the stronger statement is true because at time~$m$ process~$i$ cannot tell whether it is in~$r$ or in any of the runs $r'\ind{i}{m}r$. When this condition holds, we say that~$i$ {\em knows} that all values are~1.  More generally, it will be convenient to define the dual of indistinguishability, i.e., what is true at all indistinguishable runs, as what the process knows. More formally, following in the spirit of \cite{HM1,FHMV}, we proceed to define knowledge in our distributed systems as follows.%
\footnote{We introduce just enough of the theory of knowledge to support the analysis performed in this paper. For more details, see \cite{FHMV}.}

%
%

We shall focus on knowledge of facts that correspond to properties of runs. We call these {\em facts about the run}, and denote them by $\varphi$, $\psi$, etc. In particular, facts such as ``$v_j=1$'' ($j$'s initial value is~1), {\em ``$j$ is faulty''} (in the current run), and {\em ``$j$ decides $\commit$''} (in the current run), are all examples of facts about the run. 



\begin{definition}[Knowledge]
\label{def:know}
	Fix a system $R$, a run $r\in R$, a process~$i$ and a fact~$\varphi$.  
	We say that  $K_i\varphi$ (which we read as ``process~$i$ \defemph{knows}~$\varphi$'') holds at time~$m$ in~$r$ iff~$\varphi$ is true of all runs~$r'\in R$ such that $r'\ind{i}{m}r$.
\end{definition}

Notice that knowledge is defined with respect to a given system~$R$. Often, the system is clear from context and is not stated explicitly. 
\cref{def:know} immediately implies the so-called {\em Knowledge property}: If $K_i\varphi$ holds at (any) time~$m$ in~$r$, then $r$ satisfies $\varphi$. 


We use Boolean operators such as $\neg$ (Not), $\wedge$ (And), and~$\vee$ (Or) freely in the sequel. While the basic facts $\varphi$, $\psi$, etc. of interest are properties of the run, knowledge can change over time. Thus,  for example, $K_i(v_j=0)$ may be false at time~$m$ in a run~$r$ and true at time~$m+1$, based perhaps on  messages that~$i$ does or does not receive in round~$m+1$.

A  recent paper~\cite{kop} provides an essential connection between knowledge and action in distributed protocols called the \defemph{knowledge of preconditions principle} (KoP). It states that whatever must be true when a particular action is performed by a process~$i$ must be known by~$i$ when the action is performed. This is one way of capturing the role of indistinguishability discussed above. 
More formally, we say that a fact~$\varphi$ is a \defemph{necessary condition} for an action~$\act$ in a system~$R$ if for all runs $r\in R$ and times~$m$, if~$\act$ is performed at time~$m$ in~$r$ then~$\varphi$ must be true at time~$m$ in~$r$. For deterministic protocols in synchronous models such as~$\modelf$, the KoP 
can be stated as follows: 
\begin{theorem}[KoP, \cite{kop}]
	\label{thm:kop}
	Fix a protocol $\cal P$ for~$\modelf$, let~$i\in\Proc$ and let~$\act$ be an action of~$i$ in~$\RP$. 
	If $\varphi$ is a necessary condition for~$\act$ in~$\RP$ then 
	$K_i\varphi$ is a necessary condition for~$\act$ in~$\RP$.
\end{theorem}

We denote the fact that all initial values are~1 by $\AllOne$. Since $\AllOne$ is a necessary condition for performing $\commit$ in AC, \Cref{thm:kop} immediately implies the following, which will be useful in our analysis of AC in \cref{sec:AC}:
\begin{corollary}[Hadzilacos \cite{hadzilacos1987knowledge}]
\label{cor:know-1}
When a process~$i$  commits in a run of an AC protocol, it must know $\AllOne$ (in particular, $K_i(v_j=1)$ must hold for every process~$j$). 
\end{corollary}

\noindent
Hadzilacos \cite{hadzilacos1987knowledge} is a very elegant paper that was the first to apply knowledge theory to the study of commitment problems. 

\section{Silence in the Presence of Crashes}
\label{sec:Silence}
\label{sec:inference}

In a reliable setting, we can intuitively describe information transfer via silence as follows (we formalize the rules \Sone---\Sthree\/ in \cref{sec:silent broadcast}):
	\\[.5ex]
	\noindent\Sone.~{\bf Silent inference from a reliable source:}\quad 
 	Suppose 	that a process~$j$ is guaranteed to send~$i$ a message at time~$t$ in case some condition~$\varphi$ of interest {\em is false}. If, by time $t+\Delta$, process~$i$ has not received a message that~$j$ sent to~$i$ at time~$t$, then~$i$ can infer that~$\varphi$ holds. 
	\\[.5ex]
	In fact, extensive use of silence can achieve an even greater effect. 
	Suppose that the proper operation of a synchronous system depends on the truth of local conditions~$\varphi _j$, for different processes~$j$. In this case, all nodes of the system can be informed that all conditions hold simply if no process reports  a violation of its local condition. Effectively, this achieves a (silent) broadcast from each process to all processes, at no cost. 
	
	Notice that silence can be informative only in case there are alternative circumstances (e.g., when~$\varphi$ does not hold) under which a message would be received by~$i$. Thus, in a precise sense, the use of silence can serve to shift the communication costs among scenarios. This can be  especially useful for optimizing the behavior of protocols in cases of interest. Most popular is optimization for the common case (often referred to as  {\em fast path/slow path} protocols) in which a protocol is designed  to be very efficient in the common case, at the expense of being less efficient in uncommon cases. 
	In \Sone, the ability to convey information by silence depends on a reliability assumption, since if~$j$ may fail by crashing, for example, then it would be possible for~$i$ not to hear from~$j$ even if~$\varphi$ is untrue.  The information conveyed by a null message in this case can be described by:%
	\\[.5ex]
	\noindent\Stwo.{\bf~Silent inference from an unreliable source:}\quad
	Suppose that a process~$j$ is guaranteed to send~$i$ a message at time~$t$ in case some condition~$\varphi$ of interest {\em does not} hold and it is not faulty. 
	If, by time $t+\Delta$, process~$i$ has not received a message that~$j$ sent to~$i$ at time~$t$,
	then $i$ can infer that 
	at least one of the following is true:  (i) $\varphi$ holds, or (ii) a failure has occurred. \\[.5ex]
	In particular, if~$j$ sends messages to~$i$ over a reliable communication link, then (ii) reduces to process~$j$ being faulty. 
	Despite the fact that~$i$ may fall short of inferring that~$\varphi$ holds when it does not hear from~$j$, the information that it does obtain may still be very useful.  In many problems having to do with fault-tolerance, the specification depends on whether failures occur.
	Thus, for example, in the original {\em atomic commit} problem \cite{dwork1983inherent,guerraoui2017fast}, deciding to $\abort$ is allowed if one of the initial values is~0, or if a failure occurs. 
	In {\em Byzantine agreement}~\cite{PSL} and  {\em weak agreement}~\cite{LF82}, deciding on a default value is allowed in case of a failure. But this is not true in general. 
	In the popular {\em consensus} problem, for example, the Validity condition is independent of whether failures occur~\cite{Consensus}, thus rendering conditional knowledge as provided by \Stwo\ insufficient. 
	As we will see, there are circumstances under which silence can provide unconditional information even in the face of failures.
	
	For ease of exposition, our investigation will focus on synchronous message-passing systems, in which processes communicate in rounds, and a message sent in a particular round is received by the end of the round. 
	Moreover, we will consider settings in which only processes may fail, and there is an {\em a priori} bound of~$f$ on the number of processes that can fail. 
	One way to infer unconditional information in such settings is captured by:\\[.5ex]
	\noindent\Sthree.{\bf~Silent inference with bounded failures:}\quad
	Assume that at most~$f$ processes can fail, and
	let~$S$ be a set of at least~$f+1$ processes. 
	Moreover, suppose that every process $j\in S$  is guaranteed to send~$i$ a message no later~than round~$m$  in case~$\varphi$ {\em does not} hold and~$j$ is not faulty. If~$i$ 
	hears from no 
	process in~$S$ by time~$m$, then~$i$ \mbox{can infer that~$\varphi$  holds.} \\[.5ex]
	\Sthree\ is used, for example, by Amdur {\em et al.}~\cite{amdur1992message}, who consider message-efficient solutions to Byzantine agreement. In order to prove to all processes that the original sender completed the first round, their protocol has the sender send messages to a set~$B$ of~$f$ processes in the second round. In the third round, a member of~$B$ sends no messages if it received the sender's second round message (implying that the sender did not fail in the first round),
	and sends messages to everyone otherwise.
	The set~$S=\{\mathrm{sender}\}\cup B$ satisfies the above rule, and allows all processes to infer the fact~$\varphi=$ {\em ``the sender completed the first round successfully''} based on silence in the third round.  
	
	The rules \Sone-\Sthree\   provide sufficient conditions for silent information transfer. As we will show, 
	both \Stwo\/ and \Sthree\/ can be used to facilitate the design and analysis of efficient fault-tolerant protocols. Moreover, \Sthree\  is in a precise sense also a necessary condition for inferring particular facts of interest from silence. This, in turn, 
	will allow us to obtain simple and intuitive proofs of lower bounds and matching optimal protocols. 


\subsection{Silent Broadcast}
\label{sec:silent broadcast}
There are several ways to formally state the rules \Sone---\Sthree; we will state them in a form that will be most convenient to use in the sequel. It is natural to capture the fact that process~$i$ infers that~$\varphi$ holds using the knowledge terminology. Moreover, in a synchronous setting, not sending a message at a given point is a choice just like performing an action. In \Sone, process~$j$ is assumed not to send a message in case~$\varphi$ holds. Thus, by the knowledge of preconditions principle (\cref{thm:kop}), $j$ must know~$\varphi$ to make this choice. We say that a protocol guarantees a property if this property is true in all runs of the protocol. We can thus formalize~\Sone\ as follows: 
\begin{lemma}
	\label{lem:sBC1}
	Let~$R$ be a synchronous system with reliable communication in which process~$j$ reliably follows a protocol~$\Prot$.
	Let~$\varphi$ be a fact about runs in~$R$ and assume that $\Prot$ guarantees that if $\neg K_j\varphi$ at time~$m-1$ then~$j$ sends a message to~$i$ in round~$m$. 
	If~$i$ does not receive a message from~$j$ in round~$m$ of a run~$r\in R$ then~$K_i\varphi$ holds at time~$m$ in~$r$.
\end{lemma}
\begin{proof}
Fix $r\in R$ and assume that~$i$ does not receive a message from~$j$ in round~$m$ of~$r$. Moreover, choose an arbitrary run $r'\in R$ such that $r'\ind{i}{m}r$. By \Cref{def:know} it suffices to show that $r'$ satisfies~$\varphi$. 
By assumption, $i$ has the same local state at time~$m$ in both runs, and so it does not receive a message from~$j$ in round~$m$ of~$r'$. Since~$R$ is synchronous and communication is assumed to be reliable, it follows that~$j$ did not send a message to~$i$ in round~$m$ of~$r'$. Process~$j$ is assumed to reliably follow the protocol~$\Prot$. Hence,  by assumption, $K_j\varphi$ must be true at time~$m-1$ in~$r'$, or else it would have sent a message to~$i$. By the knowledge property it follows that~$\varphi$ holds at time~$m-1$ in~$r'$. Moreover, since~$\varphi$ is a fact about runs, $r'$ satisfies~$\varphi$, and we are done.
\end{proof}

\Stwo\ considers information gained by silence in systems in which processes can fail, such as the context~$\modelf$. In this case, $i$ might not receive a message from~$j$ due to a failure, 
 rather than because of information that~$j$ has at the beginning of the round. 
 We thus obtain a weaker variant of \Cref{lem:sBC1}, formalizing  \Stwo:


\begin{lemma}
	\label{lem:sBC2}
	Let~$\varphi$ be a fact about runs in the system~$\RP=R(\Prot,\modelf)$, and fix $i,j\in \Proc$ and a time $m> 0$. 
	Moreover, assume that $\Prot$ guarantees that if $\neg K_j\varphi$ at time~$m-1$ then~$j$ sends a message to~$i$ in round~$m$.
	If~$i$ does not receive a round~$m$ message from~$j$ in a run~$r\in\RP$ then~$K_i(\varphi \vee \text{$j$ is faulty})$ holds at time~$m$ in~$r$.
\end{lemma}


\begin{proof}
	The proof is very similar to that of \Cref{lem:sBC1}.
	Fix $r\in R$ and assume that~$i$ does not receive a message from~$j$ in round~$m$ of~$r$. Moreover, choose an arbitrary run $r'\in R$ such that $r'\ind{i}{m}r$. By \Cref{def:know} we need to show that $r'$ satisfies~$(\varphi \vee \text{$j$ is faulty})$. 
	By assumption, $i$ has the same local state at time~$m$ in both runs, and so it does not receive a message from~$j$ in round~$m$ of~$r'$.
	However, since in~$\RP$ the protocol~$\Prot$ is applied in the unreliable context~$\modelf$, there are two possible reasons for this. One is that~$j$ fails in~$r'$ before it has a chance to send~$i$ a message in round~$m$, in which case $j$ is faulty in~$r'$. The other one is that~$j$ reliably follows~$\Prot$ at time~$m-1$, in which case by assumption, $K_j\varphi$ must be true at time~$m-1$ in~$r'$ (implying~$\varphi$ is true in~$r'$), or else it would have sent a message to~$i$.
	It follows that $(\varphi \vee \text{$j$ is faulty})$ must be true in~$r'$, and we are done.
\end{proof}

Clearly, the knowledge that is obtained by silence in~$\modelf\!$, as captured by \Cref{lem:sBC2}, is 
contingent, and hence quite a bit weaker than what can be guaranteed in the more reliable setting of~\Cref{lem:sBC1}. As mentioned in the Introduction, there are problems in which such weaker knowledge may be sufficient. A closer inspection shows, however, that it is possible to obtain unqualified knowledge even in fault-tolerant settings such as the context~$\modelf\!$. The key to this is the fact that even when any single process may be faulty, a process is often guaranteed that a set of processes must contain a nonfaulty process. In this case, we can use \Cref{lem:sBC2} to obtain the following formalization of~\Sthree: 

\begin{corollary}
\label{cor:s>f}
Let $\RP=R(\Prot,\modelf)$, let~$\varphi$ be a fact about runs in~$\RP$, fix a process $i\in \Proc$ and 
a set~$|S|>f$  of processes.
Moreover, assume that for all $j\in S$ the protocol~$\cal P$ guarantees that 
if~$\neg K_j\varphi$ at time~\mbox{$m-1\ge 0$} then~$j$ sends a message to~$i$ in round~$m$.	
If, in some run $r\in\RP$, process~$i$ does not receive a round~$m$ message from any process in~$S$, then~$K_i\varphi$ holds at time~$m$ in~$r$.
\end{corollary}
\begin{proof}
Under the assumptions of the claim, choose a run $r'\ind{i}{m}r$ in~$\RP$. It suffices to show that~$\varphi$ is true in~$r'$. 
Since $|S|>f$  there must be a process $h\in S$ that is correct in~$r'$. \cref{lem:sBC2} implies that $(\varphi\vee \text{$h$ is faulty})$ is true in~$r'$. But since~$h$ is not faulty in~$r'$ we obtain that~$\varphi$ holds in~$r'$, as required.
\end{proof}

\cref{cor:s>f} shows how a set~$|S|>f$ of processes can inform a specific process~$i$ of a fact~$\varphi$ by not sending messages to~$i$. In fact, with no extra cost the set~$S$ can be silent to all processes in a round~$m$ of interest, thereby informing everyone that~$\varphi$ holds. We call this a \defemph{silent broadcast}. 


\subsection{The Silent Choir Theorem} 
\label{sec:structure}

The rule \Sthree\ and its formalization in \cref{cor:s>f} show that a protocol can orchestrate information transfer through silence by first informing (perhaps directly) a large enough set of processes about the fact of interest. We now show that, in a precise sense, this is indeed necessary. Intuitively, transferring information through silence requires enough processes to be mum to ensure that at least one of them is correct. We will call this the {\em Silent Choir Theorem}. The direct information transfer in the reliable case (\Sone) is a particular instance of this, because a choir of one suffices when there are no failures.
While the sufficient conditions of \cref{sec:silent broadcast} apply rather broadly to general facts~$\varphi$ about runs, to prove necessary conditions we need to restrict attention to certain {\em primitive} facts. This is because  knowledge of a composite fact such as conjunction $\varphi=\psi_1\wedge\psi_2$, for example, can be obtained by separately learning about each of its components. 
Intuitively, a primitive fact is a fact that is local to a process~$j$ and whose truth is not determined by the local states of the other processes. 
For ease of exposition, our analysis will exclusively consider a single primitive fact, the fact $v_j=1$ stating that~$j$'s local value is~1. Knowledge about initial values plays an important role in many problems, including consensus \cite{Consensus,DM} and AC.

Our treatment will make use of message chains. 
Given processes $i,i'\in\Proc$ and times $m,m'\ge 0$, we say that \defemph{there is a message chain from $(i,m)$ to $(i',m')$ in a run~$r$}, and write $(i,m)\mcc{r}(i',m')$, if either (a) $i=i'$ and $m\le m'$, or if (b)
there exist processes $i=j_0,j_1,\ldots,j_\ell=i'$ and times $m\le t_0<t_1,\ldots,<t_\ell\le m'$ such that 
for all $0\le k<\ell$, 
a message  from~$j_k$ to~$j_{k+1}$ is sent in~$r$ at time~$t_k$. Such a message chain is said to have \defemph{length}~$\ell$. (Notice that $m\le t_0$, so the first message in the chain may be sent after time~$m$.)  We say that there is a message chain from~$i$ to~$j$ in~$r$ if $(i,0)\mccr(j,m)$ for some time $m\ge 0$.

Recall that in~$\modelf$ process~$j$'s initial value is not known to processes $i\ne j$ at the start. Clearly, one way in which process~$i$ can come to know that $v_j=1$ is by direct communication, i.e., via a message chain from~$j$ to~$i$.
Alternatively,  $i$ can learn this fact using silence. 
The ways by which a process~$i$ can come to know that $v_j=1$ in~$\modelf$ are characterized by the following theorem. 
\begin{theorem}[Silent Choir]\label{thm:silent}\label{thm:choir}
	Let~$r$ be a run of a protocol~$\cal P$  in~$\modelf\!$. Denote by~$\bbF^r$  the set of faulty processes in~$r$,  and define $\Sr_j(t)\eqdef\{h\in\Proc\!:(j,0)\mccr(h,t)\}$. If $K_i(v_j=1)$ holds at time~$m$ in~$r$ and~$r$ does not contain a message chain from~$j$ to~$i$, then $m>0$ and \mbox{$|\bbF^r\cup \Sr_j(m-1)|>f$}. 
\end{theorem}
\begin{proof}
	Let $i$, $j$, $r$ and~$m$ satisfy the assumptions. If there is a message chain from~$j$ to~$i$ in~$r$, then we are done. So assume that no such message chain exists. Let $r'$ be a run in which $v_i$ is the same as in~$r$ and $v_j=0$. Clearly, $r\ind{i}{0}r'$ and so $\neg K_i(v_j=1)$ at time~0 in~$r$. It follows that $m>0$. 
	Roughly speaking, $\Sr_j(t)$ is the set of processes that could have learned~$v_j$ by time~$t$, via a message chain from~$j$. 
	We will establish that \mbox{$|\bbF^r\cup \Sr_j(m-1)|>f$} by proving the contraposition; i.e., that if $|\bbF^r\cup \Sr_j(m-1)|\le f$ then~$i$ does not know that $v_j=1$ at time~$m$. 
	Intuitively, the claim will follow because if $|\bbF^r\cup \Sr_j(m-1)|\le f$ then~$i$ would consider it possible at time~$m$ that~$v_j=0$ and~$j$ crashed initially, and that everyone who may have noticed its value crashed without informing anyone. 
	
	More formally, for every $h\in \Sr_j(m-1)$, define $k_h$ to be the round in~$r$ in which~$h$ first receives a message completing a message chain from~$j$. 
	Consider a run~$r'$ of~$\cal P$ in which $v_j=0$ and all other initial values are the same as in~$r$. Every process $h\in \Sr_j(m-1)$ crashes in round $k_h+1$ of~$r'$ without sending any messages. (In particular, $j$ crashes initially.) 
	The faulty processes in~$r'$ are those in $\bbF^r\cup \Sr_j(m-1)$. By assumption, their number does not exceed~$f$, and so~$r'$ is a run of~$\cal P$ in~$\modelf$. 
	
	We now prove by induction on~$\ell$ in the range $0\le\ell\le m-1$ that for all~$z\in\Proc$, if $z\notin \Sr_j(m-1)$ or $z\in \Sr_j(m-1)$ but $\ell<k_z$, then $r_z(\ell)=r'_z(\ell)$,
	and $z$ sends the same messages and performs the same actions in round~$\ell+1$ of both runs.  Recall that $j\in \Sr_j(m-1)$ and $k_j=0$. 
	For~$\ell=0$ we have that $r_z(0)=r'_z(0)$ for all $z\ne j$ since by choice of~$r'$  all initial values other than $v_j$ are the same in both runs. Since~$\cal P$ is deterministic, all processes other than~$j$ send the same messages and perform the same actions in round~1. 
	Let $0< \ell \le m-1$, and assume that the inductive claim holds for $\ell-1$.  Suppose that either $z\notin \Sr_j(m-1)$ or both $z\in \Sr_j(m-1)$ and $\ell<k_z$. Then by definition of~$\Sr_j(m-1)$ and~$k_z$ it follows that in round~$\ell$ of~$r$ process~$z$ receives no messages from a process $y\in \Sr_j(m-1)$ for which $k_y<\ell-1$. 
	Consequently, $z$ receives the same set of messages in round~$\ell$ of both runs.  Since $r_z(\ell-1)=r'_z(\ell-1)$ by the inductive hypothesis, it follows that $r_z(\ell)=r'_z(\ell)$. Again, since~$\cal P$ is deterministic, $z$ will send the same messages and perform the same actions in round~$\ell+1$ of both runs. This completes the inductive argument.

	The inductive claim implies, in particular, that all $z\notin \Sr_j(m-1)$ send the same messages in round~$m$ of both runs~$r$ and~$r'$. 
	Notice that $i\notin \Sr_j(m-1)$ by assumption, and~$i$ receives no messages from processes in~$\Sr_j(m-1)$ in round~$m$ of the original run~$r$. It follows that~$i$ receives the same set of messages in round~$m$ of both runs. Since $i\notin \Sr_j(m-1)$ we also have that $r_i(m-1)=r'_i(m-1)$, and so we obtain that $r_i(m)=r'_i(m)$. In other words, the two runs are indistinguishable to~$i$ at time~$m$. By construction, $v_j\ne 1$ in~$r'$, and so $K_i(v_j=1)$ does {\em not} hold at time~$m$ in~$r$. The claim follows. 
\end{proof}

By \cref{thm:silent}, any run in which a process learns about initial values without an explicit message chain must contain a silent choir. Indeed, it is possible to identify a silent choir in the message-optimal AC protocol of ~\cite{guerraoui2017fast}, as well as in message-optimal protocols for Byzantine agreementf~\cite{amdur1992message,hadzilacos1993message} and for failure discovery~\cite{FD}.

\section{Efficient~Protocols~for~Atomic Commitment}
\label{sec:Applying}
\label{sec:AC}
Guerraoui and Wang~\cite{guerraoui2017fast} present state-of-the-art protocols for AC optimized for nice runs in a variety of models and point out that there is a tradeoff between time and messages in their protocols. Thus, for example, for $\modelf$ they present three different protocols. In nice runs of their most message-efficient protocol,  $n+f-1$ messages are sent, and processes decide after $n+2f$ rounds. In nice runs of their fastest protocol, decisions are obtained within two rounds, and $2n^2-2n$ messages are sent.


\subsection{Lower Bounds on AC in the Common Case}\label{sec:lb}
We are now ready to apply the formalism of~\cref{sec:Silence} to the analysis of  AC protocols that are efficient in the common case. Our first goal is to show that silence must be used in message-optimal AC protocols. To do so, we start by proving a useful combinatorial property relating message chains and message complexity: 
\begin{lemma}
\label{lem:n+k-1}
Let $k>0$, and assume that the run~$r$ contains, for every process~$j\in\Proc$, message chains to at least~$k$ other processes. Then at least $n+k-1$ messages are sent in~$r$. 
\end{lemma}
\begin{proof}
First notice that every process~$i\in\Proc$ must send at least one message in~$r$, since it has message chains to other processes.
Define the \defemph{rank} of a process $i\in\Proc$ to be the length of the longest message chain from $(i,0)$ in the run~$r$. 
If there is a process with rank $\ge n+k-1$ we are done. Otherwise let~$h$ be a process with minimal rank, and let~$M$ be the set of messages on chains that start at~$h$. Clearly $|M|\ge k$,  because $r$ contains message chains from~$h$ to at least~$k$ other processes. Moreover, every $j\ne h$ must send at least one message that is not in~$M$, since otherwise~$j$'s rank would be strictly smaller than that of~$h$. Thus,  at least $n-1+|M|\,\ge\, n+k-1$ messages must be sent in~$r$, as claimed. 
\end{proof}

The message-optimal AC protocol for $\modelf$ presented in \cite{guerraoui2017fast}, 
which we shall refer to as the {\em GW protocol}, 
sends $n+f-1$ messages in nice runs.  
This can be used to prove
\begin{corollary}\label{cor:necessity}
Message-optimal AC protocols in~$\modelf$ must make use of silence when $f<n-1$. 
\end{corollary} 
\begin{proof}
Every process~$i$ commits in a nice run~$r$, and by \cref{cor:know-1} $i$~must know that $v_j=1$, for all~$j$, when it commits. \cref{thm:silent} implies that if silence is not used in~$r$, then there must be a message chain from~$j$ to~$i$. The conditions of \cref{lem:n+k-1} thus hold in~$r$ for $k=n-1$, implying that a nice run must send at least $2n-2>n+f-1$ messages. The claim follows. 
\end{proof}

Recall that the common case in this setting consists of nice runs, i.e., runs where all initial values are~1, and no failures occur. \cref{cor:know-1} captures a basic property of AC protocols: A committing process must know that all initial values are~1.
The Silent Choir Theorem and the KoP can be used to show the following property, which will be key to our lower bound proofs:
\begin{lemma}
\label{lem:know-2}
When a process~$i$  commits in a run of an AC protocol, it must know, for every $j\in\Proc$, that 
there is a message chain from~$j$ to some correct process in the current run.
\end{lemma}
\begin{proof}
Assume that~$i$ commits in the run~$r$ of~$\Prot$. By \cref{thm:kop}, it suffices to show that if~$i$ commits in~$r$ then $r$ contains a message chain from~$j$ to some correct process, for every $j\in\Proc$. Fix~$j\in\Proc$. Since $f<n$, there must be at least one correct process in~$r$; denote it by~$h$. Given that~$i$ commits in~$r$, we have by Agreement that~$h$ commits in~$r$. Let us denote by~$m$ the decision time of~$h$.  By \cref{cor:know-1} we have that $K_h(v_j=1)$ holds when~$h$ commits. If there is a message chain from~$j$ to~$h$ in~$r$ we are done, since~$h$ is the desired correct process. Otherwise, \cref{thm:silent} states that \mbox{$|\bbF^r\cup \Sr_j(m-1)|>f$}.
Thus, $\Sr_j(m-1)$ contains a correct process and there is a message chain from~$j$ to a correct process in~$r$.
\end{proof}

We are now in a position to state and prove a set of lower bounds on the tradeoff between the decision time and number of messages sent in nice runs of AC protocols. 

\begin{theorem}
\label{thm:MT-tradeoff}
Let~$\Prot$ be an AC protocol  for $\modelf$, let $D$ be the number of rounds at the end of which decision is reached in its nice runs, and let~$M$ be the number of messages sent in nice rounds. Then $D$ and~$M$ satisfy the following constraints: \vspace{-3pt}
\begin{enumerate}[itemsep=0pt]
\item[(a)] $M\ge n+f-1$;~~%
\footnote{We state and prove part (a) for completeness. It was proved for $f=n-1$ by Dwork and Skeen~\cite{dwork1983inherent}, and appears for general~$f$ in Guerraoui and Wang~\cite{guerraoui2017fast}, with some details of the proof omitted.}
\item[(b)] if $f>1$ then $D>1$;
\item[(c)] If $D=2$ then $M\ge fn$; and 
\item[(d)] If $M=n+f-1$ and $f>1$ then $D\ge 3$. 
\end{enumerate}
 \end{theorem}
\begin{proof}
Throughout the proof, $r$ will denote a nice run of an AC protocol~$\Prot$. To prove part~(a), it suffices, by \cref{lem:n+k-1}, to prove that
 the run~$r$ must contain, for every process~$j\in\Proc$, message chains from~$j$ to at least~$f$ other processes. 
Fix a process~$j$. If $r$ contains message chains from~$j$ to all other processes then we are done, since $f\le n-1$. Otherwise, let~$i$ be a process to whom there is no message chain from~$j$ in~$r$. Since~$r$ is nice, all processes commit in~$r$. Denote by~$m$ the time at which~$i$ commits. Given that there is no message chain from~$j$ to~$i$ in~$r$, \cref{thm:silent} implies that \mbox{$|\bbF^r\cup \Sr_j(m-1)|>f$}, where 
$\Sr_j(m-1)$ is the set of processes to whom a message chain from~$j$ is completed by time~$m-1$ in~$r$. Since~$r$ is nice we have that $\bbF^r=\emptyset$, and so $|\Sr_j(m-1)|>f$. Finally, as $j\in|\Sr_j(m-1)|$, there must be message chains from~$j$ to at least~$f$ other processes in~$r$. 
For part~(b), suppose that~$f>1$ and that~$i\in\Proc$ commits at time~1 in~$r$, and choose a process $j\ne i$. Consider the run~$r'$ of~$\Prot$ in which everyone starts with~1, and only~$i$ and~$j$ fail. Moreover, $j$ fails in the first round without sending messages to any process other than~$i$, while~$i$ crashes after deciding to commit, and sends no messages after the first round. The run~$r'$ is a run of~$\Prot$ in~$\modelf$, since $f>1$. The run~$r'$ contains no message chain from~$j$ to a correct process. Recall that~$i$ commits at time~1 in~$r$. 
By construction, $r'\ind{i}{1}r$, and hence, when~$i$ commits in the nice run~$r$, it does not know that the run contains a message chain from~$j$ to a correct process, contradicting \cref{lem:know-2}.
For part (c) it suffices to show that if $D=2$ then every process must send at least~$f$ messages in~$r$. 
Let $\Prot$ be an AC protocol with $D=2$, let $i\in\Proc$, and denote by~$T$ the set of processes to which~$i$  sends messages in the nice run~$r$. Assume, by way of contradiction, that $|T|<f$. Consider a run~$r'$ of~$\Prot$ in which~$i$ crashes after deciding, without sending any messages after the second round.
Moreover, every process $h\in T$ fails in the second round without sending messages to any process except possibly to~$i$ (to whom~$h$ sends iff it does in~$r$). The run~$r'$ is a run of~$\Prot$ in~$\modelf$ because it contains $|T|+1\le f$ failures. The run~$r'$ contains no message chain from~$i$ itself to a correct process. Moreover, $r'\ind{i}{2}r$ by construction, and therefore when it commits in~$r$, process~$i$ does not know that the run contains a message chain from~$j$ to a correct process. As in part (b), this contradicts \cref{lem:know-2}.
Finally, part~(d) follows directly from part~(c), since $fn$ messages must be delivered when deciding in $D=2$ rounds, and the assumption that $f>1$ implies that $fn\ge 2n>n+f-1$.
\end{proof}

Next we turn to prove that the bounds of \cref{thm:MT-tradeoff} are tight by providing matching upper bounds.

\subsection{Upper Bounds on AC in the Common Case}\label{sec:ub}
 Recall that we consider the set of processes to be $\Proc=\{0,1,\ldots,n-1\}$. 
 Roughly speaking, the message-optimal GW protocol spends $n+f-1$ rounds creating a long chain of $n+f-1$ messages sent in cyclic order from process~$1$ through~$n-1$ and then continuing to $0,\ldots,f$.  It then runs a consensus protocol for another~$f$ rounds starting at time $n+f$ to decide among the actions $\commit$ and $\abort$. In our terminology, the first process to know that $v_j=1$ for all~$j$ is process~$0$. The processes $0,1,\ldots,f$ form the silent choir: They must inform the processes in case the chain is broken. At time $n+f$, if a process is not informed of a problem, it knows that all initial values are~1. 

Using the terminology of \cref{thm:MT-tradeoff}, the GW protocol decides in $D=n+2f$ rounds.  
Dwork and Skeen's  message-optimal protocol for the case $f=n-1$ decides in $D=2n+1$ rounds \cite{dwork1983inherent}. There is a very large gap between these decision times and the lower bound of $D\ge 3$ established in \cref{thm:MT-tradeoff}(d). We now present the \Stealth\ protocol, a message-optimal AC protocol that decides in $D=3$ rounds. We will then proceed to present two additional protocols that are round optimal, deciding in $D=2$ rounds for the case of $f>1$ and in $D=1$ round when $f=1$, respectively. Together, these three protocols prove that the bounds in \cref{thm:MT-tradeoff} are all tight, and completely characterize the tradeoff between time and message complexities in the common case for AC protocols.%
\footnote{Guerraoui and Wang present a protocol called 1NBAC in \cite{guerraoui2017fast} in which processes decide in the second round in nice runs, but do so after the sending phase of the second round, and before receiving second-round messages. 1NBAC requires $2n^2-2n$ messages in nice runs. In \cref{sec:1.5} we prove a lower bound of $n^2+fn-n$ messages for  their model, and present a matching protocol that proves that the bound is tight.}
Full descriptions of all protocols in this section, and their proofs of correctness, appear in Appendices \ref{sec:QC} -- \ref{sec:D1f1}.

\vspace{2mm}

\noindent\underline{\Stealth:\hspace{1mm}\/ an AC protocol with $D=3$ and $M=n+f-1$}.$~$\\[.7ex]
{\bf Fast Path:}~~In the first round, all processes that have value~1 send a message to process~$0$. If process~$0$ hears from all processes, it sends messages to processes $1,2,\ldots,f$ in the second round. In case all these messages are indeed sent, no messages are sent in the third round. A process that receives no message in the third round performs $\commit$ at time~3. 
 Finally, if such a process receives no messages in the fourth round, it halts at time~~4.\\
{\bf Slow Path:}~~ The slow path in \Stealth\ operates as follows.
Any member of the choir that did not receive its second round message (or $n-1$ first-round messages in the case of process~$0$) broadcasts (i.e., sends all processes)  an ``error'' message in the third round. Any process that receives error messages in the third round broadcasts a clarification request message ``huh?'' in the fourth round. Finally, if any ``huh?'' message is sent in the fourth round, then the processes run an $f-1$ crash-tolerant consensus for~$f$ rounds starting in the fifth round. The variant of consensus used by \Stealth\ has the property that if a correct process starts with value~1 then the final decision is~1. For details, see \cref{sec:Stealth}. 
A process votes~1 (in favor of committing)  in the consensus protocol if it received no message in the third round, or if it is in the choir and received its second round message. It votes~0 (in favor of aborting) otherwise. Processes that participate in the consensus protocol halt when their role in the consensus protocol ends. In the worst case, the last process halts in \Stealth\ at time $f+5$.\hfill $\clubsuit$

\vspace{3mm}
In nice runs, the \Stealth\/ protocol constructs a silent choir for the fact that all values are~1 in round~2, and uses round~3 to allow the choir to perform a silent broadcast of this fact to all processes. 
Interestingly, it uses silence in an additional manner: Round~4 performs a silent broadcast from all processes to all processes. This does not require a silent choir because it makes use of the rule \Stwo: a null message from~$j$ informs~$i$ that $j$ is either active and has committed, or has crashed without deciding. 
The next  protocol also uses two rounds of silence in a similar manner. 
\vspace{2mm}

\noindent\underline{\Quick:\hspace{1mm}\/ an AC protocol with $D=2$ and $M=fn$.}$~$\\[.7ex]
{\bf Fast Path:}~~In the first round, every process~$j$ with $v_j=1$ sends messages to the $f$ processes $j+1,j+2,\ldots,j+f$ (wrapping around mod~$n$). A process that sent messages and received all~$f$ possible messages in the first round remains silent in the second round. A process that receives no messages in the second round performs $\commit$ at time~2. Finally, if such a process receives no messages in the third round, it halts at time~3.
\\{\bf Slow Path:} The slow path is similar in spirit to that of~\Stealth, starting one round earlier. A process~$j$  such that $v_j=0$, or that does not receive first-round messages from all of its~$f$ predecessors $j-1,j-2,\ldots,j-f$ (mod~$n$) broadcasts an ``error'' message in the second round. In the third round, a process that (sent or) received an ``error'' message in the second round broadcasts a message listing all processes whose values it knows to be~1. 
All processes that do not halt at time~3 participate in a consensus protocol as in \Stealth, that starts in the fourth round. In this case, a process enters the consensus protocol with value~1 iff it knows that all initial values were~1. \hfill$\clubsuit$
\vspace{3mm}

Our third protocol treats the boundary case of $f=1$. Rather than two silent rounds, only the latter one, corresponding to \Stwo, is needed in nice runs: 
\vspace{2mm}

\noindent\underline{\Short:\hspace{1mm}\/ an AC protocol with $D=1$ and $M=n^2-n\,$ for $f=1$.}$~$\\[.8ex]
{\bf Fast Path:}~~In the first round, every process~$j$ with $v_j=1$ sends messages to all other processes. 
 A process with value~1 that receives messages from all processes in the first round performs $\commit$ at time~1. Finally,  if such a process receives no messages in the second round, it halts at time~2.
\\{\bf Slow Path:} For the slow path, a process~$j$ with $v_j=1$ that does not receive messages from everyone in the first round broadcasts a  ``huh?'' message in the second round. A committed process that receives such a message in the second round, responds in the third round with an ``$\AllOne$'' message to the requesting process~$j$. Finally,  a process that did not commit at time~1 will perform $\commit$ at time~3 if it received an ``$\AllOne$'' message in the second round, and will perform $\abort$ otherwise. All processes halt  at time~3 (possibly after deciding at time~3) and send no messages in the fourth round. \hfill$\clubsuit$


\section{Discussion}
\label{sec:discussion}
In the well-known story {\em Silver Blaze}, Sherlock Holmes is able to conclude that a victim's dog was familiar with the murderer, based on the fact that the dog was silent when the criminal entered \cite{Blaze}. In that example, Conan Doyle crisply illustrates how crucial knowledge can be gained by observing silence.
The focus of our investigation has been on the use of silence in distributed protocols. 
A variety of protocols in the literature make implicit use of silence (e.g., \cite{CGM,dolev1983authenticated,lenzen2016near}). 
 In network protocols,  periodically sending {\it Keep-alive} messages is commonly used to inform a peer that a node is active. The absence of such messages allows the peer to discover that a failure has occurred \cite{KeepAlive}. 

As has been elucidated by  \cite{lamport1984using}, synchronous channels allow processes to send null messages by not sending a message. 
In a precise sense, our analysis considers what information can be passed using null messages. 
This becomes especially interesting in systems with failures, since the absence of a message in such settings can be caused by a failure, rather than by deliberate silence. 
In the presence of a bounded number of crash failures, we identified the notion of a silent choir as an essential component of information transfer by silence. We showed that the only way that a process~$j$ can inform another process~$i$ of~$j$'s initial value without an explicit message chain between them is by constructing a silent choir (\cref{thm:silent}). 
Indeed, in a variety of message-optimal protocols it is possible to identify silent choirs. These include 
 \cite{guerraoui2017fast} for Atomic Commitment, \cite{amdur1992message,hadzilacos1993message} for Byzantine Agreement, and \cite{FD}  for the failure discovery problem.

The silent choir is a valuable tool for a protocol designer. 
By directly constructing silent choirs, we were able to reduce the number of rounds for message-optimal AC protocols from $n+2f$ to~3 in the \Stealth\ protocol, and to design a protocol (\Quick) that is message optimal among  round-optimal AC protocols in \cref{sec:ub}. 
Since silent choirs are necessary in some cases, they can also serve in theoretical analyses. In \cref{thm:MT-tradeoff}, we used the necessity of silent choirs to obtain a complete characterization of the tradeoff between decision times and communication costs for AC protocols.

A notion of potential causality for reliable synchronous systems based on allowing null messages in message chains was defined in \cite{centipede}. Defining a similar notion for fault-prone systems by using silent choirs rather than null messages is an interesting topic for further investigation. 

While our investigation focused on the synchronous round-based model $\modelf$ in which processes fail by crashing and the network is fully connected, it applies  to more general synchronous systems. In particular, the rules \Stwo\/ and \Sthree\/ and their formalizations, as well as the Silent Choir Theorem, hold with minor modifications even in the presence of Byzantine failures. The use of silence, and more generally the topic of how information can be gleaned from actions that are {\em not} performed by processes in a system, is an exciting topic that deserves further investigation in the future.


\bibliographystyle{plain}
\bibliography{z}

\appendix


\section{The \Stealth\/ Protocol}\label{sec:QC}\label{sec:Stealth}
Recall that the goal is optimizing AC for the common-case, i.e., for nice runs. Therefore, in order to capture the essence of silence and simplify the protocols, improvements regarding other cases (the uncommon cases) are left out.
In some not ``nice'' runs the \Stealth\/ and \Quick\/ protocols make use of a consensus variant we denote as~$\BB$ (standing for {\sc Biased-to-1-Uniform-Consensus}). The properties that a $\BB$ protocol satisfies are:
\begin{definition}\label{def:1-Cons NEW}
	Protocol $\cal P$ is a {\sc Biased-to-1-Uniform-Consensus} (which we denote by~$\BB$), iff every run $r$ of $\cal {P}$ satisfies the following conditions:
		\begin{itemize}[itemsep=.5pt]
			\item[]\textbf{Validity:}\quad A process decides on a value $\hat{v}$ only if some process proposes $\hat{v}$.
			\item[]\textbf{Agreement:}\quad  No two process decide differently.
			\item[]\textbf{Decision:}\quad Every correct process eventually decides, \quad and
			\item[]\textbf{Biased-to-1:}\quad If some correct process starts with~1, processes may only decide~1.
		\end{itemize}
\end{definition}
\noindent The \Stealth\/ protocol, sketched in \cref{sec:ub}, is presented in \cref{fig:Stleath-portocol}. 
It reaches the lower bound of~\cref{thm:MT-tradeoff}(d). That is, for~$f>1$, \Stealth\/ is round-optimal among message-optimal AC protocols with~$M=n+f-1$ and~$D=3$.
The protocol is also optimal in bits as well as in messages.
Each message sent in \Stealth\/ consists of a single bit.
Both~$\AllOne$,~`\err' and~`huh?' can be represented by~`1' as they are the only messages that might appear in rounds~2, ~3 and~4 respectively.

\begin{figure}[h] 
	\centering{ 
		\fbox{ 
			\begin{minipage}[t]{150mm} 
				\footnotesize 
				\renewcommand{\baselinestretch}{2.5} 
				\begin{tabbing} 
					aaaaaaaa\=aa\=aaaaaaa\=\kill  
					{\bf round 1} \\					
					$\forall i\in \Proc$:\\
					\>\>{\bf if} $v_i=1$ {\bf then} send `1' to process~$0$ 
					\\[1ex]
					{\bf round 2}\\ 					
					Process~$0$:\\			
					\>\>{\bf if} received `1' from all {\bf then} send $\AllOne$ to processes $\left\{0,1,2,...,f\right\}$
					\\[1ex]{\bf round 3}\\
					Processes $\left\{0,1,2,...,f\right\}$:\\					
					\>\>{\bf if} 
					received no message from process~$0$ in round~2 {\bf then} send `\err' to all
					\\[1ex]{\bf round 4}\\
					$\forall i\in \Proc$:\\
					\>\>{\bf if} did not receive any `\err' {\bf then} $\commit$\\
					\>\> {\bf else} send `huh?' to all		
					\\[1ex]{\bf round 5} \\
					$\forall i\in \Proc$:\\					
					\>\> {\bf if} no  `huh?' message received {\bf then} halt\\[1ex]
					\>\>{\bf else}\\[1ex]
					\>\>	\hskip1cm		$\;\hat{v}_i \leftarrow \left\{
					\begin{tabular}{ l l }
					1 & {\bf if} received $\AllOne$ or performed $\commit$\\
					0 & otherwise \\
					\end{tabular} \right\};$\\[1ex]
					\>\>\hskip1cm  $\decision\gets \BB(f-1,\hat{v}_i)$;\\[1ex]
					\>\>\hskip1cm {\bf if} haven't performed~$\commit$ {\bf then}\\[.7ex]
					\>\>\hskip1cm\begin{tabular}{ l l }
						\hskip.8cm$\commit$ & {\bf if} $\decision=1$,\\
						\hskip.8cm$\abort$ & {\bf if} $\decision=0$ \\
					\end{tabular}
				\end{tabbing} 
				\normalsize 
			\end{minipage} 
		} 
		\caption{\Stealth\/: protocol by rounds.
			All processes participate in rounds~1,4 and~5, only process~$0$ in round~2, and only processes~$\left\{0,1,2,...,f\right\}$ in round~3.
			In nice runs everyone decides in round~4 and halts in round~5.}
		\label{fig:Stleath-portocol} 
	} 
\end{figure}

\begin{figure}[!b]
	\centering
\begin{subfigure}[h]{.4\textwidth} 
	\centering{ 
		\includegraphics[width=\textwidth 
		]{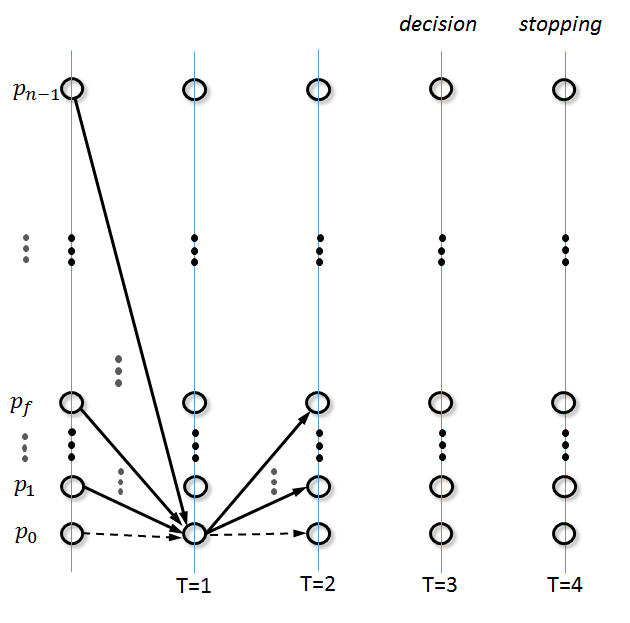}
		\caption{\Stealth\/: A nice run (all vote~1 and no failures occur).}
		\label{fig:QC-nice} 
	} 
\end{subfigure}\hspace{.5cm}%
\begin{subfigure}[h]{0.5\textwidth} 
	\centering{ 	
		\includegraphics[width=\textwidth 
		]{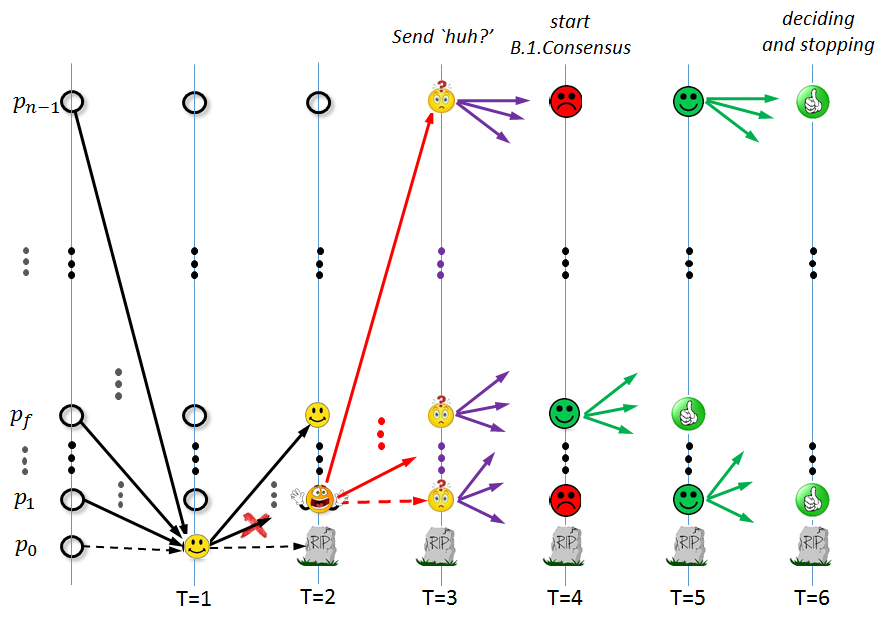}
		\caption{\Stealth\/: Example of a run with failures.}
		\label{fig:QC-not_nice} 
	} 
\end{subfigure}
\caption{Example for two possible runs of \Stealth\/}
\label{fig:exmQC}
\end{figure}

Two possible runs of \Stealth\ are illustrated in \cref{fig:exmQC}.
A nice run goes as in \cref{fig:QC-nice}: \textbf{round~1} all processes send~1 to process~$0$; \textbf{round~2} process~$0$ sends $\AllOne$ to processes~$1,2,..,f$; \textbf{round~3} all processes are quiet; \textbf{round~4} all processes decide~1; and finally in \textbf{round~5} all processes halt.
\cref{fig:QC-not_nice} depicts an example of a complicated run that is not nice. It goes as follows:
in \textbf{round~1} all processes send~1 to process~$0$; in \textbf{round~2} process~$0$ tries to send $\AllOne$ to processes~$1,2,..,f$ but crashes. some receive it and some not; in \textbf{round~3} processes from~$\{ 1,2,..,f\}$ that didn't receive $\AllOne$, send~\err\ to all, the rest are quiet; in \textbf{round~4} all processes that received~\err\ send~`huh?' to all (asking for help); in \textbf{round~5} the living processes start $\BB$, a process that knows $\AllOne$ starts with~1 (green smiley), the rest start with~0 (red smiley).


\subsection{Correctness of \Stealth}\label{sec:Correctness of Stealth}
Recall that the specification of the AC problem is given in~\cref{sec:Ac-def}. 
The detailed description of \Stealth\/ appears in \cref{fig:Stleath-portocol} below. 
We now consider the protocol's correctness. Throughout the Appendix, the protocols are assumed to execute in the context~$\modelf$.

\begin{claim}\label{clm:QCComplexity}
	If~$r$ is a nice run of \Stealth , then exactly $n+f-1$ messages are sent over the network, and all processes perform~$\commit$ at time~3 in~$r$.
\end{claim}
\begin{proof}
	Let~$r$ be a nice run of \Stealth .
	Since~$r$ is nice, it holds that $v_i=1$ for all~$i\in \Proc$ and no failures occur, therefore,~$r$ goes as follows:
	In the first round $n-1$ messages of~`1' are sent to process~$0$. In~round 2, process~$0$ sends~$\AllOne$ to $S=\left\{0,1,2,...,f\right\}$, adding~$f$ to the total number of sent messages (it does not actually sends a message to itself).
	In round~3 all is quiet (no message is sent), then, at time~3, all processes perform~$\commit$.
	This results in a total of~$n+f-1$ sent messages. The communication pattern is illustrated in~\cref{fig:QC-nice}.
\end{proof}

\begin{lemma}\label{lem:commitValidity}
	All runs of \Stealth\/ satisfy \textit{commit validity}.
\end{lemma}
\begin{proof}
	Let~$r$ be a run of \Stealth\/ in which some process, say~$i$, starts with~0, we show that no decision in~$r$ is~$\commit$.
	By construction of \Stealth\/, since~$i$ starts with~0 in~$r$, it does not send~`1' to process~$0$ in round~1. Thus, in round~2 of~$r$, process~$0$ does not receive~`1' from all and will not send~$\AllOne$ to anyone.
	In round~3 of~$r$, none of processes in~$S=\left\{0,1,2,...,f\right\}$ receive~$\AllOne$ from process~$0$, so they send~\err\ to all (if not crashed).
	Since $\left| S \right|=f+1$, it is guaranteed that at least one of them is correct in~$r$, so in round~3 all have heard \err\ (at least from the correct process). Thus, in round~4 of~$r$, no one commits and each process send `huh?' to all.
	In round~5 $\BB$ is initialized and starts.
	Since no process ever received $\AllOne$ from process~$0$ in the 2nd round of~$r$, all processes in $\BB$ are initialized with~0.
	From $\BB$ validity guarantee, it only returns~0 in~$r$, hence, the only decisions possible in~$r$ is~$\abort$ and no process performs~$\commit$ in~$r$.
\end{proof}

\begin{lemma}\label{lem:abortValidity}
	All runs of \Stealth\/ satisfy \textit{abort validity}.
\end{lemma}

\begin{proof}
	Let~$r$ be a run of \Stealth\/ in which no process starts with~0 (they all start with~1) and no failures occur, then~$r$ is a nice run by definition. By~\cref{clm:QCComplexity} no process performs~$\abort$ in~$r$ (they all perform~$\commit$).
\end{proof}

\begin{lemma}\label{lem:decision}
	All runs of \Stealth\/ satisfy \textit{decision}.
\end{lemma}

\begin{proof}
	Let~$r$ be a run of \Stealth\/ and let~$i\in \Proc$ be a correct process in~$r$.
	If~$i$ does not decide in round~4 of~$r$, then it starts $\BB$.
	By the decision guarantee of $\BB$,~$i$ decides in a finite number of rounds in~$r$.
\end{proof}

\begin{lemma}\label{lem:agreement}
	All runs of \Stealth\/ satisfy \textit{agreement}.
\end{lemma}

\begin{proof}
	Let~$r$ be a run of \Stealth\/.	
	We prove by dividing into the two possible cases for~$r$:	
	~\textbf{case (1)} All decisions in~$r$ are done after round~4. By construction of \Stealth\/ this implies that all decisions in~$r$ are made according to a consensus protocol, therefore, in this case~$r$ satisfies agreement based on the agreement property pf consensus.
	~\textbf{case (2)} Some process, say~$i$, decides by round~4 in~$r$. 
	By construction of \Stealth\/ every process that decides by round~4, including~$i$, decides~$\commit$ in round~4.
	For process~$i$ to commit in round~4 of~$r$, all process in $S=\left\{0,1,2,...,f\right\}$ must be quiet in round~3.
	$\left| S \right|=f+1$, therefore, at least one of the process in~$S$ is correct in~$r$ and is quiet because it knows~$\AllOne$ (process~$0$ or someone that heard $\AllOne$ from process~$0$).
	Let's name this process $p_{\mbox{\tiny good}}$.
	By construction of \Stealth\/, processes that decide \textbf{after} round~4 in~$r$ send~`huh?' to all in round~4. After which all processes that have not crashed by round~5 in~$r$ start~$\BB$.
	$p_{\mbox{\tiny good}}$ is correct in~$r$ and therefore if a process decides round~4 in~$r$ it participates in~$\BB$, and since~$p_{\mbox{\tiny good}}$ knows~$\AllOne$, it initializes its consensus' initial-value to~1.
	$\BB$ is a uniform-consensus-biased-to-1. Thus, because of $p_{\mbox{\tiny good}}$ is correct and starts with~1, the biased to-1 condition of ensures that all decisions of~$\BB$ in~$r$ will be~1, which leads to~$\commit$. Consequently, every process in~$r$ that decides after round~4 decides~$\commit$ (like all decisions by round~4 in~$r$), and~$r$ satisfies agreement in this case as well.
	We got that in both possible cases~$r$ satisfies agreement.	
\end{proof}

\begin{claim}\label{clm:QCsolvesAC}
	\Stealth\/ is an AC protocol in $\modelf$.
\end{claim}
\begin{proof}
	Lemmas~\ref{lem:commitValidity},~\ref{lem:abortValidity},~\ref{lem:decision}, and~\ref{lem:agreement} prove that every run of the \Stealth\/ protocol in the context~$\modelf$ satisfies the required conditions for AC, thus, establishing that \Stealth\/ is an AC protocol in the context of~$\modelf$.
\end{proof}

\section{The \Quick\/ Protocol}\label{sec:2MD}
The protocol in~\cref{fig:2.Rounds-algorithm} commits at time~2 in nice runs and uses only~$nf$ messages to do so.
\begin{figure}[h] 
	\centering{ 
		\fbox{ 
			\begin{minipage}[h]{150mm} 
				\footnotesize 
				\renewcommand{\baselinestretch}{2.5} 
				\setcounter{linecounter}{0} 
				\begin{tabbing} 
					aaaaaaaa\=aa\=aaaaaaa\=\kill  
					$\forall i\in\Proc$:\\					
					{\bf round 1} \\
					\>{\bf if} $v_i=1$ {\bf then} send `1' to $\{i,i+1,...,i+f\}\mod n$;
					\\[1ex]
					{\bf round 2}\\ 					
					\> {\bf if} received fewer than~$f+1$ messages of `1'  {\bf then} send `\err ' to all;
					\\[1ex]
					{\bf round 3}\\ 					
					\> {\bf if} did not receive `\err ' in previous round {\bf then} $\commit$;\\
					\> {\bf else} send everyone the id's of processes from whom 1 was received in round 1;
					\\[1ex]
					{\bf round 4}\\ 					
					\> {\bf if} did not receive any messages in round~3 {\bf then} halt;\\[.7ex]
					\>{\bf else}\\[1ex]
					\>\>	\hskip0.9cm		$\;\hat{v}_i \leftarrow \left\{
					\begin{tabular}{ l l }
					1 & {\bf if} received all id's $\{0,1,...,n-1\} $ or performed $\commit$\\
					0 & otherwise \\
					\end{tabular} \right\};$\\[1ex]
					\>\>\hskip1cm  $\decision\gets \BB(f-1,\hat{v}_i)$;\\[1ex]
					\>\>\hskip1cm {\bf if} haven't performed~$\commit$ {\bf then}\\[.7ex]
					\>\>\hskip1cm\begin{tabular}{ l l }
						\hskip.8cm$\commit$ & {\bf if} $\decision=1$,\\
						\hskip.8cm$\abort$ & {\bf if} $\decision=0$ \\
					\end{tabular}
				\end{tabbing} 
				\normalsize 
			\end{minipage} 
		} 
		\caption{\Quick\/ - Protocol: Total cost of a nice run = $nf$.}
		\label{fig:2.Rounds-algorithm}
	} 
\end{figure} 

\noindent We next prove the correctness of \Quick.

\subsection{Correctness of \Quick }\label{sec:Correctness of Quick}
\begin{claim}\label{clm:D2complexity}
	If~$r$ is a nice run of~~$\Prot _{D2}$ in context~$\modelf$, then exactly $fn$ messages are sent over the network and all processes perform~$\commit$ at time~2 in~$r$.
\end{claim}
\begin{proof}
	Let~$r\in R(\Prot _{D2},\modelf)$ be a run in which all processes start with~1 and no failure occurs (a nice run). Then~$r$ goes as follows:	
	In round~1 of~$r$ every~$i\in \Proc$ sends~1 to~$\{ i,i+1,...,i+f \} $ (a total of $n\cdot f$ messages), and they all receive the message.
	Thus, in round~2 all processes received the expected messages so they all remain quiet.
	Since every process~$i$ did not receive~\err\ by time~2, it commits at time~2.
	Hence, exactly $fn$ messages are sent and all processes commit at time~2 in~$r$.
\end{proof}

\begin{lemma}\label{lem:2comVal}
	All runs of $\Prot _{D2}$ in context~$\modelf$ satisfy \textit{commit validity}.
\end{lemma}
\begin{proof}
	Let~$r$ be a run of$\Prot _{D2}$ in~$\modelf$, in which some process $i\in \Proc$ starts with~$v_i=0$. We show that no process commits in~$r$.
	
	In round~1 $i$ does not send~1 to~$\{ i,i+1,...,i+f \} $, thus, in round~2 all correct processes in~$\{ i,i+1,...,i+f \} $ send~\err\ to all.
	Since~$|\{ i,i+1,...,i+f\} |>f$, at least one of these process is correct in~$r$ and therefore its~\err\ messages reaches all processes by time~2.
	Because every process received an~\err\ message in round~2, in round~3 none of them decides. (Instead, they all send the id's of processes from whom they received~1.)
	As no process received~1 from~$i$ in round~1, no one sends the id of~$i$ in round~3, Hence, no process receives all id's by time~4 in~$r$.
	Therefore, in round~4 every~$j\in \Proc $ starts consensus with~$\hat{v}_j=0$. By consensus validity no process decides~1 (and consequently commits) in~$r$.
\end{proof}

\begin{lemma}\label{lem:2abrVal}
	All runs of $\Prot _{D2}$ in context~$\modelf$ satisfy \textit{abort validity}.
\end{lemma}
\begin{proof}
	Let~$r$ be a run of $\Prot _{D2}$ in~$\modelf$, in which all processes start with~1 and no failure occur (a nice run). We show that all processes commit at time~2 in~$r$.
	
	In round~1 of~$r$ every~$i$ sends~1 to~$\{ i,i+1,...,i+f \} $, and they all receive the message.
	Thus, in round~2 all processes received the expected messages so they all remain quiet.
	Since every process~$i$ did not receive~\err by time~2, it commits (in the beginning of round~3).
	Hence, all processes commit in~$r$.
\end{proof}

\begin{lemma}\label{lem:2agree}
	All runs of $\Prot _{D2}$ in context~$\modelf$ satisfy \textit{agreement}.
\end{lemma}
\begin{proof}
	Let~$r$ be a run of $\Prot _{D2}$ in~$\modelf$. If all decisions in~$r$ are done before round~4, it is only committing, and thus, $r$ satisfies agreement.
	If no decision in~$r$ is done before round~4, then all decisions are according to a uniform consensus protocol and thus,~$r$ satisfies agreement.
	Hence, the only case left to check is when some, but not all, processes decide before round~4, specifically this happens at time~2.
	
	Let~$i$ be a process that decides at time~2 in~$r$, we will prove that all other processes that decide in~$r$, decide the same as~$i$.
	By construction of $\Prot _{D2}$, every decision before round~4 is commit, thus,~$i$ and whomever decides before round~4 are committing. We are left to show that every decision after time~4 is to commit as well.
	Since~$i$ did not receive any \err\ messages in round~2 and no more then~$f$ failures are possible in~$r$, then, for every~$j\in \Proc$ there is at least one correct process in~$\{ j,j+1,...,j+f \} $ which received~1 from~$j$ in round~1. This process succeeds in sending~$j$'s id to all in round~3 of~$r$.
	In round~4, for every~$j\in \Proc$ its id was received by all processes. Therefore, all processes start the consensus protocol with initial value of~1, and by validity of consensus all decisions must be~1 ( and consequently to $\commit$). Thus, they all agree with~$i$, proving the claim that~$r$ satisfies agreement.
\end{proof}

\begin{lemma}\label{lem:2dec}
	All runs of the $\Prot _{D2}$ in context~$\modelf$ satisfy \textit{decision}.
\end{lemma}
\begin{proof}
	Let~$r$ be a run of $\Prot _{D2}$ in~$\modelf$. By construction of~$\Prot _{D2}$, correct processes who don't decide at time~2 or~$r$ will do so using the consensus protocol that starts in round~4. As the consensus protocol satisfies decision, all correct processes eventually decide in~$r$.
\end{proof}

\begin{claim}\label{clm:2DsolvesAC}
	Protocol $\Prot _{D2}$ (\cref{fig:2.Rounds-algorithm}) is an AC protocol in $\modelf$
\end{claim}
\begin{proof}
	Lemmas~\ref{lem:2comVal},~\ref{lem:2abrVal},~\ref{lem:2agree}, and~\ref{lem:2dec} prove that all runs of $\Prot _{D2}$ in the context~$\modelf$ satisfy the required conditions for AC, thus, making protocol~\Quick\/ an AC protocol in the context of~$\modelf$.
\end{proof}

\section{The \Short\ Protocol}\label{sec:D1f1}
In the special case of a single possible failure protocol~\Short\ is presented in~\cref{fig:D1f1-algorithm}.
\begin{figure}[h] 
	\centering{ 
		\fbox{ 
			\begin{minipage}[t]{150mm} 
				\footnotesize 
				\renewcommand{\baselinestretch}{2.5} 
				\setcounter{linecounter}{0} 
				\begin{tabbing} 
					aaaaaaaa\=aa\=aaaaaaa\=\kill  
					$\forall i\in \Proc$:\\					
					{\bf round 1} \\
					\>{\bf if} $v_i=1$ {\bf then} send `1' to all; 
					\\[1ex]
					{\bf round 2}\\ 					
					\> {\bf if} $v_i=1$ and received `1' from all in previous round {\bf then} $\commit$;\\
					\> {\bf else} send `huh?' to all;
					\\[1ex]
					{\bf round 3}\\ 					
					\> {\bf if} performed $\commit$ and did not receive any `huh?' in previous round {\bf then} halt;\\
					\> {\bf else if} performed $\commit$ and received `huh?' from~$j$ in previous round {\bf then} send $\AllOne$ to~$j$;
					\\[1ex]
					{\bf round 4}\\ 					
					\> {\bf if} received $\AllOne$ in previous round {\bf then} $\commit$; halt;\\
					\> {\bf else } $\abort$; halt;
				\end{tabbing} 
				\normalsize 
			\end{minipage} 
		} 
			\caption{\Short\/ - Protocol: A special protocol for~$f=1$ that decides at time~1 (For~$f>1$ we proved~$D\ge 2$). Total cost of a nice run is~$n^2-n$ messages.}
		\label{fig:D1f1-algorithm} 
	} 
\end{figure} 

\subsection{Correctness of \Short}

\begin{claim}\label{clm:D1f1Complexity}
	If~$r$ is a nice run of~~$\Prot_{D1f1}$ in~$\gamma^1$, then exactly~$n^2-n$ messages are sent over the network and all processes perform~$\commit$ at time~1 in~$r$.
\end{claim}
\begin{proof}
	Let~$r\in R(\Prot_{D1f1},\gamma^1)$ be a run in which all processes start with~1 and no failure occur (a nice run). Then~$r$ goes as follows:	
	In round~1 of~$r$ every process sends~1 to all, and they all receive the message. Thereafter, since every process had received~1 from all, it performs~$\commit$ at time~1 of~$r$.
	As all perform $\commit$ in round~2 no~`huh?' message is sent in~$r$. Finally, all processes halt at time~2.
	Hence, every process sends~$n-1$ messages (a total of~$n^2-n$ over the network) and commits at time~1 in~$r$.
\end{proof}

\begin{corollary}\label{lem:D1f1abVal}
	All runs of~~$\Prot_{D1f1}$ in~$\gamma^1$ satisfy \textit{abort validity}.
\end{corollary}
\begin{proof}
	Let~$r$ be a run of~~$\Prot_{D1f1}$ in~$\gamma^1$, in which all processes start with~1 and no failures occur (a nice run). Then by~\cref{clm:D1f1Complexity} no process aborts in~$r$ (they all commit).
\end{proof}

\begin{lemma}\label{lem:D1f1comVal}
	All runs of~~$\Prot_{D1f1}$ in~$\gamma^1$ satisfy \textit{commit validity}.
\end{lemma}
\begin{proof}
	Let~$r\in R(\Prot_{D1f1},\gamma^1)$ be a run in which some process~$j\in \Proc$ starts with~$v_j=0$. We show that no process commits in~$r$.	
	In round~1 of~$r$,~$j$ does not send~1 to~anyone, thus, in round~2 all processes did not receive~1 from all and therefore all processes send~`huh?' to all.
	Since no process performed~$\commit$ in round~2, no $\AllOne$ messages are sent in round~3 of~$r$.
	Finally, a process that reaches round~4 has received no message of~$\AllOne$ previously, therefore, it decides~$\abort$ at time~4 of~$r$.
	Hence, no process commits in~$r$ (they either abort or crash).
\end{proof}


\begin{lemma}\label{lem:D1f1agreement}
	All runs of~~$\Prot_{D1f1}$ in~$\gamma^1$ satisfy \textit{agreement}.
\end{lemma}
\begin{proof}
	Let~$r$ be a run of~~$\Prot_{D1f1}$ in~$\gamma^1$.
	If no decision (to $\commit$) is done at time~1 in~$r$, then there is no process to send~$\AllOne$, therefore, decisions are made only to~$\abort $ at time~3 of~$r$, and thus, $r$~satisfies agreement.
	On the second case, if a process, say~$i$, decides~$\commit$ at time~1 in~$r$, then if~$i$ does not crash until round~4 in~$r$ it would send~$\AllOne$ to all undecided processes in round~3, thus, enforcing them to perform~$\commit$ at time~3 of~$r$.
	If, however,~$i$ is faulty in~$r$ then no other process fails in~$r$ (no more than a single failure is possible in~$\gamma^1$).  Hence, since~$i$ decided in round~2, $\forall j\in \Proc : v_j=1$ and no one crashes in round~1. Thus, all decide to~$\commit$ at time~1 in~$r$.
	We get that on the second case, some process decides~$\commit$ at time~1, then all decisions are~$\commit$ and~$r$ satisfies agreement.
	Thus, for all cases	$r$~satisfies agreement.
\end{proof}

\begin{lemma}\label{lem:D1f1dec}
	All runs of~~$\Prot_{D1f1}$ in~$\gamma^1$ satisfy \textit{decision}.
\end{lemma}
\begin{proof}
	Let~$r$ be a run of~~$\Prot_{D1f1}$ in~$\gamma^1$.
	By construction of~~$\Prot_{D1f1}$, correct processes who don't decide at time~1 in~$r$, decide at time~3 in~$r$. Therefore, all correct processes eventually decide in~$r$ (at time~3 latest).
\end{proof}

\begin{claim}\label{clm:D1f1solvesAC}
	Protocol~~$\Prot_{D1f1}$ (\cref{fig:D1f1-algorithm}) is an AC protocol in~$\gamma^1$.
\end{claim}
\begin{proof}
	Lemmas~\ref{lem:D1f1abVal},~\ref{lem:D1f1comVal},~\ref{lem:D1f1agreement}, and~\ref{lem:D1f1dec} prove that all runs of protocol~~$\Prot _{D1f1}$ in the context~$\gamma^1$ satisfy the required conditions for AC, thus, making protocol~~$\Prot _{D1f1}$ an AC protocol in the context of~$\gamma^1$.
\end{proof}


\section{Committing in Mid-Round}\label{sec:1.5}
As proved in~\cref{thm:MT-tradeoff}(b), if~$f>1$ then no AC protocol commits in time~1 of a nice run in~$\modelf\!$.
Guerroui and Wang however, present in~\cite{guerraoui2017fast} a protocol that solves AC and commits at time~1 of a nice run! How is this magic performed?

The answer lays in the model. In~\cite{guerraoui2017fast} the commonly used synchronous context of~$\modelf$ is replaced with a different one (which we denote as context~$\tilde{\gamma}^f$).
The difference is that in~$\tilde{\gamma}^f$ if process~$i$ performed a message sending action in its crashing round, it is guarantied that the message would be sent correctly, whereas in~$\modelf$, the sending of~$i$'s messages in its crashing round is not guarantied.
Since,~$\tilde{\gamma}^f$ is a particular case of~$\modelf$ our protocols are also correct in~$\tilde{\gamma}^f$. The opposite is however not true.
For completeness we also analyze the decision before time~2 in~$\tilde{\gamma}^f$ here.

The following lemma forms the basis of the analysis,
\begin{lemma}\label{lem:1r-msgs}
	Assume that $f>1$, let~$r$ be a nice run of an AC protocol in~$\tilde{\gamma}^f$, and let~$i\in \Proc$ be a process that decides before time~2 in~$r$. Then (a)~$i$ receives a first-round message from every process in~$r$, and 
	(b)~$i$ sends at least~$f$ round-two messages before it commits. 
\end{lemma}

\begin{proof}
	Let~$r$ and~$i$ satisfy the assumptions in the claims statement. In particular, $i$ commits in~$r$ at time~$1$, perhaps after sending some messages in round two. By~\cref{cor:know-1}, process~$i$ needs to know that all initial values were~1 when it commits. If~$i$ does not receive a round one message from some process $j\ne i$, then there is a run~$r'$ in which $v_j=0$, everyone else starts with~1, and~$j$ crashes before sending any messages. 
	Clearly~$r'$ is a legal run, $r_i(1)=r'_i(1)$, and so~$i$ does not know that $v_j=1$ at time~1 in~$r$. This establishes (a). 
	
	To prove (b), fix, in addition, a process~$j$. Since $f>1$ and~$i$ has no information about failures at time~1 in~$r$, $i$ considers it possible that~$j$ crashed without sending round one messages to anyone other than~$i$ and that, in addition, $i$ itself will crash immediately after committing.
	
	Recall from \cref{lem:know-2} that when~$i$ commits, it must know that there is a message chain in the run from~$j$ to at least one correct process, with respect to every $j\in \Proc$.  If~$i$ is the only process who has heard from~$j$ in round one and it crashes immediately after committing, then the only way that there will be a message chain from~$j$ to a correct process is if~$i$ sends round-two messages to at least $f-1$ processes other than~$j$ before it commits. Since $f>1$ by assumption, $i$ must send at least one message in round two before committing.  Let $j'$ be a process to whom~$i$ sends a round-two message before committing in~$r$. The same argument as we have made w.r.t.~$j$ (ensuring a message chain from~$j'$ to a correct process), yields that~$i$ must send round-two messages to at least $f-1$ processes other than~$j'$. Since~$i$ also sends a round-two message to~$j'$, it follows that~$i$ must send at least~$f$ round two messages before committing, establishing claim~(b). 
\end{proof}


\cref{lem:1r-msgs} states that every process that commits before time~2 in a nice run must receive at least $n-1$ messages in the first round and send at least~$f$ messages in the second.
Thus, we get:
\begin{corollary}
	\label{cor:lessThan2}
	For~$f>1$, no AC protocol in context~$\tilde{\gamma}^f$ can commit in a nice run without sending round-two messages.
\end{corollary}

Moreover, using~\cref{lem:1r-msgs} and the fact that every round-one message is received by a unique process, and every round-two message is sent by a unique process, we obtain: 
\begin{corollary}
	\label{cor:time1}
	Let~$\Prot$ be an AC protocol in~$\tilde{\gamma}^f$, in which all processes commit before time~2 in every nice run. Then at least $n^2+fn-n$ messages are sent in every nice run of~$\Prot$.
\end{corollary}
This lower bound is below the upper bound of~$2n(n-1)$ presented in~\cite{guerraoui2017fast}, thereby leaving a gap of~$n^2-(f+1)n$. We therefore close this gap with the protocol in~\cref{fig:1.5 rounds-algorithm}.

Deciding before being able to receive second-round messages makes it impossible to use silence in the second round. Therefore round one must achieve an all to all message pattern in order for $K_i \AllOne$ to hold for every $i\in \Proc$. By sending messages to~$f$ other process informing them~$\AllOne$ in the second round,~$i$ knows that a correct process will know~$\AllOne$ and thereby,~$i$ can commit.

As in \Stealth, all messages of the protocol consist of a single bit. The only non-trivial one is `huh?', implemented in 1-bit by: quiet to $\{i+1,...,i+f\}\% n$; 1 to the rest;
\begin{figure}[h] 
	\centering{ 
		\fbox{ 
			\begin{minipage}[t]{150mm} 
				\footnotesize 
				\renewcommand{\baselinestretch}{2.5} 
				\setcounter{linecounter}{0} 
				\begin{tabbing} 
					aaaaaaaa\=aa\=aaaaaaa\=\kill  
					$\forall i\in \Proc$:\\					
					{\bf round 1} \\
					\>{\bf if} $v_i=1$ {\bf then} send `1' to all; 
					\\[1ex]
					{\bf round 2}\\ 					
					\> {\bf if} received `1' from all in previous round\\
					\>\>\hskip.7cm {\bf then} send $\AllOne$ to $\{i+1,...,i+f\}\mod n$; 
					$\commit$;\\
					\> {\bf else} send `huh?' to all;
					\\[1ex]
					{\bf round 3}\\ 					
					\> {\bf if} did not receive any `huh?' in previous round {\bf then} halt;\\
					\>{\bf else}\\[1ex]
					\>\>	\hskip0.9cm		$\;\hat{v}_i \leftarrow \left\{
					\begin{tabular}{ l l }
					1 & {\bf if} received $\AllOne$ or performed $\commit$\\
					0 & otherwise \\
					\end{tabular} \right\};$\\[1ex]
					\>\>\hskip1cm  $\decision\gets \BB(f-1,\hat{v}_i)$;\\[1ex]
					\>\>\hskip1cm {\bf if} haven't performed~$\commit$ {\bf then}\\[.7ex]
					\>\>\hskip1cm\begin{tabular}{ l l }
						\hskip.8cm$\commit$ & {\bf if} $\decision=1$,\\
						\hskip.8cm$\abort$ & {\bf if} $\decision=0$ \\
					\end{tabular}
				\end{tabbing} 
				\normalsize 
			\end{minipage} 
		} 
		\caption{\textit{1.5D-algorithm}: Total cost of a nice run is $n^2-n+nf$.}
		\label{fig:1.5 rounds-algorithm} 
	} 
\end{figure} 

\subsection{Correctness of \textit{1.5D} Protocol}

\begin{claim}\label{clm:1.5Complexity}
	If~$r$ is a nice run of~~$\Prot _{1.5D}$ in~$\tilde{\gamma}^f$, then exactly~$n^2+nf-n$ messages are sent over the network and all processes perform~$\commit$ at time~1 in~$r$.
\end{claim}
\begin{proof}
	Let~$r$ be a run of~~$\Prot _{1.5D}$ in~$\tilde{\gamma}^f$, in which all processes start with~1 and no failure occur (a nice run). Then~$r$ goes as follows:	
	In round~1 of~$r$ every process sends~1 to all, and they all receive the message ($n-1$ messages per process).
	In round~2 of~$r$ every process had received~1 from all, therefore it sends its~$f$ round~2 messages successfully, and then perform~$\commit$ at time~1.
	Hence, every process sends~$n+f-1$ messages (a total of~$n^2+nf-n$ over the network) and commits at time~1 in~$r$.
\end{proof}

\begin{corollary}\label{lem:1.5abVal}
	All runs of~~$\Prot _{1.5D}$ in context~$\tilde{\gamma}^f$ satisfy \textit{abort validity}.
\end{corollary}
\begin{proof}
	Let~$r$ be a run of~~$\Prot _{1.5D}$ in~$\tilde{\gamma}^f$, in which all processes start with~1 and no failures occur (a nice run). Then by~\cref{clm:1.5Complexity} no process aborts in~$r$ (they all commit).
\end{proof}

\begin{lemma}\label{lem:1.5comVal}
	All runs of~~$\Prot _{1.5D}$ in context~$\tilde{\gamma}^f$ satisfy \textit{commit validity}.
\end{lemma}
\begin{proof}
	Let~$r$ be a run of~~$\Prot _{1.5D}$ in~$\tilde{\gamma}^f$, in which some process~$j\in \Proc$ starts with~$v_j=0$. We show that no process commits in~$r$.	
	In round~1 of~$r$,~$j$ does not send~1 to~anyone, thus, in round~2 all processes did not receive~1 from all and therefore all processes send~`huh?' to all.
	Since no $\AllOne$ messages were received in round~2 of~$r$, in round~3 processes begin the consensus protocol with~$\hat{v}=0$, therefore by validity of consensus no process decides~1 and therefore no process commits in~$r$.
\end{proof}


\begin{lemma}\label{lem:1.5agreement}
	All runs of~~$\Prot _{1.5D}$ in context~$\tilde{\gamma}^f$ satisfy \textit{agreement}.
\end{lemma}
\begin{proof}
	Let~$r$ be a run of~~$\Prot _{1.5D}$ in~$\tilde{\gamma}^f$. If all decisions are done before round~3 of~$r$, it is only committing, and thus, $r$~satisfies agreement.
	If no decision is done in~$r$ before round~3, then all decisions are according to a uniform consensus protocol and thus, $r$~satisfies agreement.
	Hence, the only case left to check is when some, but not all, processes decide before round~3 (in round~2).
	
	Let~$i$ be a process that decides before round~3 in~$r$, we will prove that all other processes that decide in~$r$, decide the same as~$i$.
	By construction of~~$\Prot _{1.5D}$ every decision before round~3 is $\commit$, thus,~$i$ and whomever decides before round~3 commit in~$r$. We are left to show that every decision after time~3 in~$r$ is $\commit$ as well.
	Since~$i$ committed in round~2 of~$r$, $K_i \AllOne$ holds in~$r$, moreover,~$i$ informed~$f+1$ processes (including itself) that~$\AllOne$ holds in~$r$, we denote this group of processes by~$S_i$.
	The bound of~$f$ possible failures insures us that at least one process in~$S_i$ is correct and knows~$\AllOne$ in~$r$. This correct process will start the consensus protocol with~$\hat{v}=1$ and by the 1-tending property of the consensus used, all decisions would be~1 which leads to~$\commit$.
	Thus, they all agree with~$i$, proving the claim that~$r$ satisfies agreement.
\end{proof}

\begin{lemma}\label{lem:1.5dec}
	All runs of~~$\Prot _{1.5D}$ in context~$\tilde{\gamma}^f$ satisfy \textit{decision}.
\end{lemma}
\begin{proof}
	Let~$r$ be a run of~~$\Prot _{1.5D}$ in~$\tilde{\gamma}^f$.
	By construction of~~$\Prot _{1.5D}$, correct processes who don't decide in round~2 of~$r$, decide using the consensus protocol that starts in round~3. As the consensus protocol satisfies decision in~$\tilde{\gamma}^f$, all correct processes eventually decide in~$r$.
\end{proof}

\begin{claim}\label{clm:1.5DsolvesAC}
	Protocol~~$\Prot _{1.5D}$ (\cref{fig:1.5 rounds-algorithm}) is an AC protocol in $\tilde{\gamma}^f$.
\end{claim}
\begin{proof}
	Lemmas~\ref{lem:1.5abVal},~\ref{lem:1.5comVal},~\ref{lem:1.5agreement}, and~\ref{lem:1.5dec} prove that all runs of protocol~~$\Prot _{1.5D}$ in the context~$\tilde{\gamma}^f$ satisfy the required conditions for AC, thus, making protocol~~$\Prot _{1.5D}$ an AC protocol in the context of~$\tilde{\gamma}^f$.
\end{proof}


\end{document}